\documentclass[letterpaper, 10 pt, conference]{ieeeconf}  

\IEEEoverridecommandlockouts                              
\usepackage{array}
\usepackage{tabularx}
\usepackage{xcolor}
\usepackage{amsmath}
\usepackage{amsfonts}
\usepackage{mathtools}
\usepackage{stmaryrd}
\usepackage[ruled,vlined]{algorithm2e}
\usepackage{tcolorbox}
\usepackage{wrapfig}
\usepackage{tikz}
\usepackage{graphicx}
\usepackage{svg}
\usepackage{amssymb}
\usepackage{hyperref}
\usepackage{moresize}
\usepackage[utf8]{inputenc}
\usepackage{lipsum}
\usepackage{newtxmath}
\let\labelindent\relax
\usepackage{enumitem}
\usepackage{scalerel}

\hypersetup{
    colorlinks=true,
    linkcolor=black,
    urlcolor=blue,
    pdftitle={},
    pdfpagemode=FullScreen,
    }
\usepackage{anyfontsize}

\usetikzlibrary{shapes.geometric, arrows}
\usetikzlibrary{calc,fit,arrows}
\newtheorem{theorem}{Theorem}[section]
\newtheorem{lemma}[theorem]{Lemma}

\newtheorem{problem}[theorem]{Problem}

\newtheorem{remark}[theorem]{Remark}
\newtheorem{assumption}[theorem]{Assumption}
\usepackage{accents}

\usepackage{cite}

\title{\textbf{Towards Certified Sim-to-Real Transfer via Stochastic Simulation-Gap Functions}}
\author{P Sangeerth$^1$, Abolfazl Lavaei$^2$, and Pushpak Jagtap$^{1}$
\thanks{$^*$ This work was supported in part by ANRF ARG Grant, ARTPARK, and Siemens.}
\thanks{$^1$Center for Cyber-Physical Systems, Indian Institute of Science, Bengaluru, India. Emails: {\tt\small \{sangeerthp,pushpak\}@iisc.ac.in}}%
\thanks{$^2$School of Computing, Newcastle University, United Kingdom. Email: {\tt\small abolfazl.lavaei@newcastle.ac.uk}}
}
  
\begin{document}
\maketitle
\pagestyle{empty}


\begin{abstract}
This paper introduces the notion of \emph{stochastic simulation-gap function}, which formally quantifies the gap between an approximate mathematical model and a high-fidelity \emph{stochastic} simulator. Since controllers designed for the mathematical model may fail in practice due to unmodeled gaps, the stochastic simulation-gap function enables the simulator to be interpreted as the nominal model with bounded state- and input-dependent disturbances.  We propose a data-driven approach and establish a formal guarantee on the quantification of this gap. Leveraging the stochastic simulation-gap function, we design a controller for the mathematical model that ensures the desired specification is satisfied in the high-fidelity simulator with high confidence, taking a step toward bridging the sim-to-real gap. We demonstrate the effectiveness of the proposed method using a TurtleBot model and a pendulum system in stochastic simulators.
\end{abstract}

\section{Introduction}
\label{sec:introduction}
Mathematical models of real-world systems often exhibit inaccuracies, rendering traditional model-based controllers unreliable. This has driven the development of data-driven control techniques \cite{ ESMAEILI2024540,10156081}, which design controllers directly from input-output data without requiring an explicit system model. Some approaches \cite{jagtap2020control, awan2023formal, 8368325} utilize Gaussian processes to provide probabilistic guarantees, while others \cite{10015033,10644811,lopez2020robust,9683437} employ data-driven control barrier functions (CBF), as well as robust and adaptive control strategies.

Collecting extensive data for data-driven control is challenging in practice. Advanced simulators address this issue by enabling precise and efficient data collection, thus improving Sim2Real transfer. Simulation platforms such as Gazebo \cite{1389727}, CARLA \cite{dosovitskiy2017carla}, ADAMS-Simulink \cite{brezina2011using}, Unreal Engine\footnotemark[3]{}, NVIDIA Drive Constellation\footnotemark[4]{}, CarMaker\footnotemark[5]{}, and CoppeliaSim\footnotemark[6]{} improve simulation realism through sensor models, physics engines, and uncertainty models, supporting the development of effective data-driven controllers.
\footnotetext[3]{https://www.unrealengine.com}
\footnotetext[4]{https://resources.nvidia.com/en-us-auto-constellation/drive-constellation}
\footnotetext[5]{https://ipg-automotive.com/en/products-solutions/software/carmaker/}
\footnotetext[6]{https://www.coppeliarobotics.com/}

Several studies, such as \cite{10260398,tan2018sim,10611401,akella2023safety}, focus on bridging the gap between simulation environments and reality within the Sim2Real context. In \cite{10611401}, a technique is developed to minimize the Sim2Real gap in object detection.
In reinforcement learning, the Sim2Real gap has been measured as the difference in expected rewards between simulated and real systems \cite{tan2018sim}, though without formal guarantees. In \cite{akella2023safety}, the Sim2Real gap is quantified by modeling the real world as a deterministic system and representing the gap as a numerical value with probabilistic guarantees.

\begin{figure}[t]
	\centering
	\includegraphics[width=0.75\linewidth]{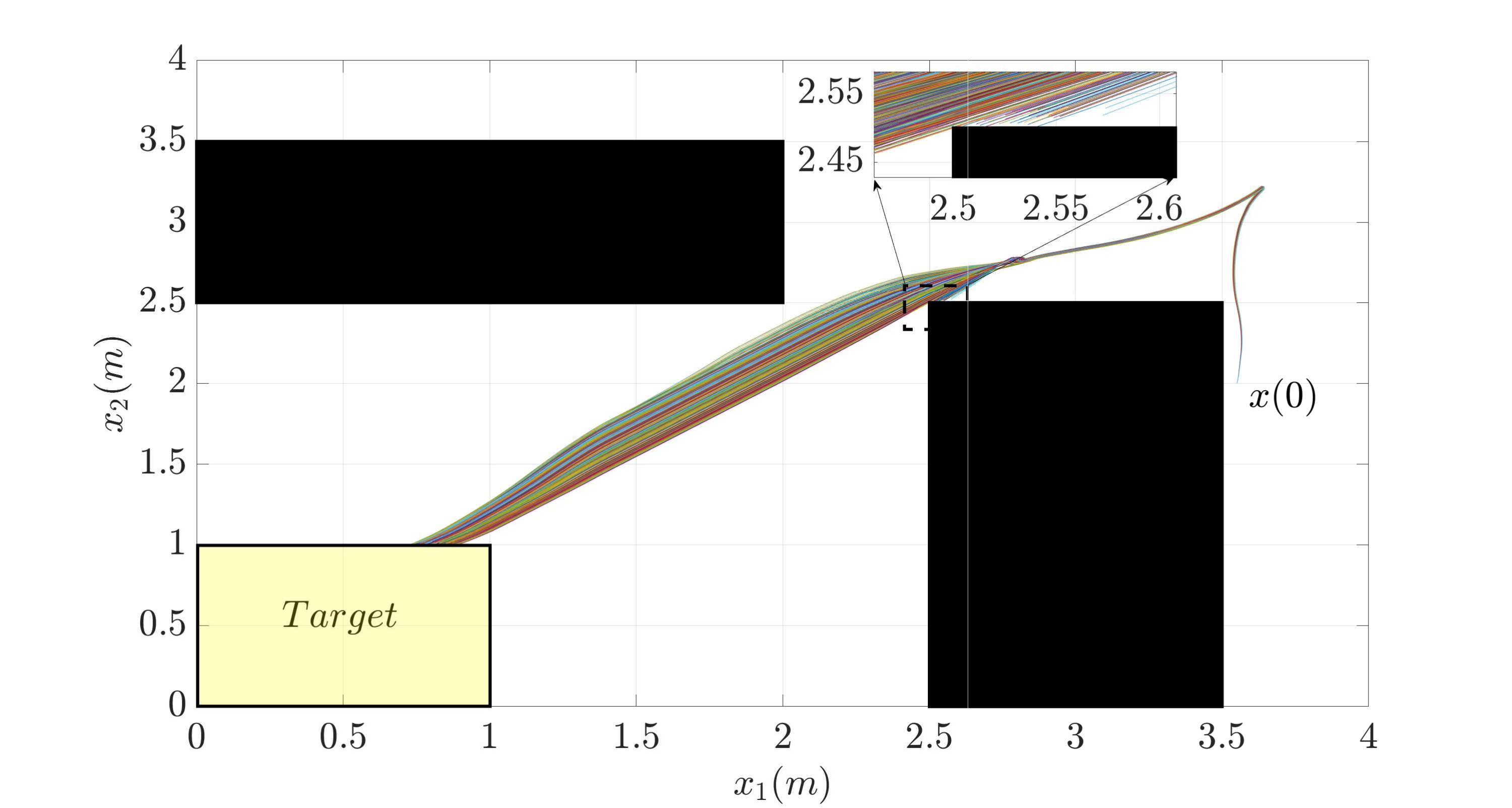}
	\caption{Simulation showing that out of 1000 realizations of TurtleBot from $(x,y,\theta) = (3.5,2,1.5)$ using the controller and deterministic simulation gap $\gamma(x,u)$ from \cite{sangeerth2024towards}, 44 resulted in collisions in the Gazebo simulator due to unmodeled simulator stochasticity. Obstacles are shown in black.}
	\label{fig:  with det gamma failure}
\end{figure}

{In \cite{sangeerth2024towards}, the Sim2Real gap is represented as a state- and input-dependent function under the assumption of a deterministic simulator. In practice, however, both simulators and real-world systems typically exhibit stochastic behavior due to factors such as sensor noise, actuator imperfections (\emph{e.g.,} delays and dead-band), and parametric uncertainties \cite{10196019}. To better capture real-world variability, high-fidelity simulators frequently incorporate stochastic models for sensing and actuation. To demonstrate this effect, we revisit Case Study II from \cite{sangeerth2024towards} and generate multiple trajectories in Gazebo starting from the same initial condition using a controller designed under deterministic assumptions, following \cite{sangeerth2024towards}. The resulting trajectories reveal noticeable variability caused by simulator stochasticity. As shown in Fig. \ref{fig: with det gamma failure}, this variability can occasionally lead to violations of the task specification (\emph{i.e.,} failure to reach the target while avoiding obstacles). These observations underscore the importance of explicitly accounting for stochasticity in high-fidelity simulators when addressing the sim-to-real gap, which forms the primary motivation for this work.}

{\bf Key Contribution.} Inspired by this critical challenge, this work adopts a data-driven framework to quantify the \emph{stochastic} simulation gap between the nominal model and the high-fidelity stochastic simulator, assuming the latter closely approximates reality. To the best of our knowledge, this is the first approach to model the sim-to-real gap as a state- and input-dependent function while accounting for simulator stochasticity. Using this gap, we employ model-based control to design a controller that ensures the desired specifications are met in the stochastic simulator with probabilistic guarantees. We validate our proposed approach on two nonlinear systems: a Turtlebot (Gazebo) and a pendulum (PyBullet).

\section{System Description}

\label{preliminaries section}
\textbf{Notations.} We denote sets of non-negative integers and non-negative real numbers, respectively, by $\mathbb{N}$ and $\mathbb{R}^+_0$. The $n$-dimensional Euclidean space is $\mathbb{R}^n$, and a column vector $x \in \mathbb{R}^n$ is written as $x = [x_1; x_2; \cdots; x_n]$ with Euclidean norm represented by $\lVert x\rVert$.  For a matrix $A \in \mathbb{R}^{m \times n}$, $A^\top$ denotes its transpose. The Borel $\sigma-$algebras on a set $X$ are denoted by $\mathcal{B}(X)$. The measurable space on $X$ is denoted by $(X,\mathcal{B}(X))$. The probability space is denoted by $(\Omega, \mathcal{B}(\Omega), \mathbb{P})$, where $\mathcal{B}(\Omega)$ is the $\sigma$-algebra, and $\mathbb{P}$ is the probability measure. Considering a random variable $z$, $Var(z):=\mathbb{E}(z^2)-(\mathbb{E}(z))^2$, where $\mathbb{E}$ denotes the expectation operator. For a set $\mathcal{C}$, its complement is represented as $\mathcal{C}^c$.\\

In this paper, we consider a nominal mathematical model represented by a discrete-time system, obtained by discretizing conventional continuous-time models \cite{rabbath2013discrete} with a sampling time $\tau > 0$, given by
\begin{equation}
\label{disc_mathematical_model}
    {\Sigma}\!: x^+ = {f}_{\tau}(x,u),
\end{equation}
where $x^+$ represents the state variables at the next time step, \emph{i.e.,} $x^+ \coloneq x(k + 1), \; k \in \mathbb{N}$.  Moreover, $x =[x_1; x_2; \cdots ; x_n] \in X \subset \mathbb{R}^n$ is the state vector in the compact state set $X$, $u=[u_1;u_2;\cdots; u_m] \in U\subset \mathbb{R}^m$ is the input vector in the \emph{finite} input set $U$ with cardinality $M$, and ${f}_{\tau}:\textit{X} \times \textit{U} \rightarrow \textit{X}$ is the known transition map, which is assumed to be locally Lipschitz continuous to guarantee the uniqueness and existence of the solution \cite{Khalil:1173048}.

Due to modeling inaccuracies, traditional mathematical model-based controllers often fail to perform effectively in the real world. However, recent advances in sensor technology allow precise state measurements, improving the system's understanding. This has facilitated the development of high-fidelity simulators such as Webots\footnotemark[7]{}, Pybullet \cite{coumans2021}, and MuJoCo \cite{6386109}, which further enhanced the ability to replicate real-world dynamics by incorporating sensor models, physics engines, and uncertainty models. 

In this work, we model the evolution of the real-world control system in stochastic simulators as an unknown stochastic discretized map with the same sampling parameter $\tau$:
\footnotetext[7]{{http://www.cyberbotics.com/}}
\begin{equation}\label{simulator_model}
    \hat{\Sigma}\!: x^+=\hat{f}_{\tau}(x,u,w),
\end{equation}
where the definitions of $x \in X$ and $u \in U$ remain the same as that of the mathematical model, while $w=[w_1; w_2; \cdots ; w_p] \in \mathbb{R}^p$ is a random variable on the Borel space ${V}_w$ with an \emph{unknown} distribution $\mathbb{P}$. Moreover, $\hat{f}_{\tau}:\textit{X} \times \textit{U} \times V_w \rightarrow \textit{X}$ is an \emph{unknown} transition map. For brevity, $f_{\tau}$ and $\hat{f}_{\tau}$ will be represented as $f$ and $\hat{f}$, respectively, throughout the paper. {The notation $f_i$ and $\hat{f}_i$ denotes the $i$-th component of the vector $f$ and $\hat{f}$.}

We now formally outline the problem under consideration in this work. 
\label{problem formulation section} 
\begin{tcolorbox}[width=\linewidth, colback=white,colframe=black]
\begin{problem}
    \label{Problem-1}
    Given the nominal mathematical model ${\Sigma}$ in \eqref{disc_mathematical_model} and the stochastic high-fidelity simulator model $\hat{\Sigma}$ in \eqref{simulator_model}, we aim to
    \begin{enumerate}
        \item[(i)] formally quantify a stochastic simulation gap between $\Sigma$ and $\hat\Sigma$, represented by a function $\gamma_i:X \times U \rightarrow {\mathbb{R}^+_0}$ such that $\forall i \in \{1,2,\dots,n\}$, $\forall x \in X$, $\forall u \in U $, and $\forall w \in V_w$,
        \begin{equation}
        \label{problem_statement_equation}
            |\mathbb{E}_w\hat{f}_i(x,u,w)-f_i(x,u)| \leq \gamma_i(x,u)
        \end{equation}
        with a probabilistic guarantee and
        \item[(ii)] design a controller using the nominal mathematical model and the stochastic simulation-gap function $\gamma(x,u):=[\gamma_1(x,u);\gamma_2(x,u);\ldots;\gamma_n(x,u)]$ that enforce the desired specifications in the stochastic high-fidelity simulator ${\hat\Sigma}$ while offering the probabilistic guarantee.
    \end{enumerate}
\end{problem}
\end{tcolorbox}
\begin{remark}
\label{gamma_triviality_remark}
    Note that an arbitrarily large $\gamma_i(x,u)$ (\emph{e.g.} trivially $+\infty$) upper bounds $|\mathbb{E}_w\hat{f}_i(x,u,w)-{{f}_i}(x,u)|$. This work aims to obtain the \emph{tighter} upper bound by minimizing $\gamma_i(x,u)$ using the proposed optimization framework. The subscript $i$ belongs to the entire set $I$ unless otherwise stated.
\end{remark}

\section{Data-driven Framework} \label{Sec:Data-driver Framework}
To solve Problem \ref{Problem-1}, we develop a data-driven approach using data collected from the mathematical model and the high-fidelity simulator to derive the stochastic simulation-gap function through an optimization approach. 

The condition \eqref{problem_statement_equation} can be reformulated as the following robust optimization program (ROP) for all $i \in I$:
\begin{align}\notag
\min_{\gamma_i \in \mathcal{H},\eta_i \in \mathbb{R}}
&\eta_i  \\\notag
\text{s.t.}\quad &\gamma_i(x,u) \leq \eta_i, \quad \forall x\!\in\! X, \forall u \!\in\! U,\\\notag
&|\mathbb{E}_w\hat{f}_i(x,u,w)\!-\!{{f}_i}(x,u)|\!-\! \gamma_i(x,u) \!\leq\! 0, \nonumber\\ \label{eq: ROP_for_one_state}
&\forall x\!\in\! X, \forall u\! \in \! U,
\end{align}
where $\mathcal{H}:=\{g \,\big|\, g\!:X \times U \rightarrow \mathbb{R}\}$ is a functional space.

Solving the ROP in \eqref{eq: ROP_for_one_state} presents significant challenges. First, it has infinite constraints due to the continuous state space. Second, $\hat{f}_i(x,u,w)$ is unknown, making it impossible to take its expectation directly. Third, the structure of $\gamma_i(x,u)$ is also unknown. To overcome these challenges, we start by fixing the structure of $\gamma_i(x,u)$ as $\sum_{l=1}^{z_i} q^{(l)}_ip^{(l)}_i(x,u)$, a parametric form linear in the decision variable $q_i= [q^{(1)}_i;\cdots;q^{(z_i)}_i] \in \mathbb{R}^{z_i}$, with user-defined (potentially nonlinear) basis functions $p^{(l)}_i(x,u)$. These basis functions can be chosen arbitrarily as any smooth functions. While prior knowledge of the system can inform this choice, the Weierstrass Approximation Theorem \cite{stone1948generalized} can also provide justification for using \emph{polynomials}, as they can uniformly approximate any continuous function over compact sets using polynomial functions.

We aim to develop a data-driven scheme for computing the simulation-gap function without directly solving the ROP in \eqref{eq: ROP_for_one_state}. To achieve this, we collect $N$ samples of the state $x_r$ from $X$, where $r=\{1,2,\ldots,N\}$. Consider a ball $X_r$ around each sample $x_r$ with radius $\epsilon$, such that for all points in $X$, there exists $x_r\in X$ that satisfies
\begin{equation}
\label{max_epsilon}
    \Vert x-x_r \Vert \leq \epsilon, \quad \forall x \in X_r.
\end{equation}
{This ensures that the collection of these balls forms a finite cover of the compact state set $X$, satisfying  $X \subseteq \bigcup_{r=1}^N X_r$.}

Recall that the cardinality of the finite input set is $M$. For all states, for each $ r\in\{1,\ldots,N\}$, with $x=x_r$, and for all $u\in U$, if one can collect $N\times M$ data of $f_i(x,u)$ and $\mathbb{E}_w\hat{f}_i(x,u,w)$ from the systems $\Sigma$ and $\hat{\Sigma}$ for the sampling time $\tau$, the ROP in \eqref{eq: ROP_for_one_state} can be reformulated as a scenario convex program (SCP) as follows:
\begin{align}
\min_{q_i \in \mathbb{R}^{z_i}, \eta_i \in \mathbb{R}} &\eta_i  \nonumber\\
\text{s.t.}\quad  & q_i^\top p_i(x_r,u) \hspace{-0.2em}\leq\hspace{-0.2em} \eta_i,~ \forall r\in\{1,\ldots,N\},u \in U, \nonumber\\
&|\mathbb{E}_w\hat{f}_i(x_r,u,w)\hspace{-0.2em}-\hspace{-0.2em}{{f}_i}(x_r,u)|-q_i^\top p_i(x_r,u)\hspace{-0.2em}\leq \hspace{-0.2em}0, \nonumber\\&\forall r\in\{1,\ldots,N\}, u \in U.\label{eq: SCP_for_one_state}
\end{align}
 The term $\mathbb{E}_w\hat{f}_i(x_r,u,w)$ in the SCP is generally unknown since the distribution $\mathbb{P}$ is not known to us. To address this, we approximate the expectation using its empirical mean by drawing $\hat{N}_1$ independent
and identically distributed (i.i.d.) samples $w_k, k\in\{1,2,\ldots,\hat{N}_1\}$, from $\mathbb{P}$ for each $(x_r,u)$ pair. In total, $N \times M \times \hat{N}_1$ data of $x^+$ is collected from $\hat{\Sigma}$ to shape the scenario convex program as follows, denoted by SCP-$\hat{N}_1$:
\begin{align}
\min_{q_i \in \mathbb{R}^{z_i}, \eta_i \in \mathbb{R}} &\eta_i  \nonumber\\
\text{s.t.} \quad  &q_i^\top p_i(x_r,u) \leq \eta_i, \forall r\in\{1,\ldots,N\}, \forall u \!\in\!U, \nonumber\\
&\bigl|\frac{1}{\hat{N}_1}\sum_{k=1}^{\hat{N}_1}\hat{f}_i(x_r,u,w_k)-{{f}_i}(x_r,u)\bigr|-q_i^\top p_i(x_r,u)\nonumber\\&+\delta_{1,i} \leq 0, \forall r\in\{1,\ldots,N\}, \forall u \in U.
\label{eq: SCP_for_one_state_n}
\end{align}
In SCP-$\hat{N}_1$, the term $|\frac{1}{\hat{N}_1}\sum_{k=1}^{\hat{N}_1}\hat{f}_i(x_r,u,w_k)-{{f}_i}(x_r,u)|$ represents the absolute difference between the empirical mean of the stochastic evolutions of simulator and the deterministic evolution of the mathematical model. This difference can be computed directly using the dataset without requiring explicit knowledge of $\hat{f}$ or the probability distribution $\mathbb{P}$. 
Additionally, in the second constraint of SCP-$\hat{N}_1$, there is an extra term $\delta_{1,i}$, compared to the corresponding constraint in SCP \eqref{eq: SCP_for_one_state}. While this term generally makes the constraint more conservative, it accounts for the error introduced by approximating the expectation with the empirical mean. This adjustment is achieved using Chebyshev inequality \cite{Saw01051984} as given in Lemma~\ref{Chebyshev bound theorem for fsim-fmath} below. 

We begin by assuming bounded variance of the function $\hat{f}_i(x,u,w)$, which is required for the results in Lemma~\ref{Chebyshev bound theorem for fsim-fmath}.

\begin{assumption}
\label{variance assumption}
We assume the existence of a bound $\hat{M}^{(i)}$, such that $Var(\hat{f}_i(x,u,w)) \leq \hat{M}^{(i)}, \forall x \in X, \forall u \in U$~\cite{salamati2024data}.
\end{assumption}
Under Assumption~\ref{variance assumption}, the following theorem provides a probabilistic guarantee that the solution of SCP-$\hat{N}_1$ is also a solution of SCP.

\begin{lemma}
\label{Chebyshev bound theorem for fsim-fmath}
    Let $[\eta_i^*;q^{(1)}_i;\cdots;q^{(z_i)}_i]$ be a feasible solution of SCP-$\hat{N}_1$ \eqref{eq: SCP_for_one_state_n}. Suppose that Assumption \ref{variance assumption} holds with a given $\hat{M}^{(i)}$. Then, for some arbitrary $\delta_{1,i} > 0$, one has
    \begin{equation}  
     \label{Chebyshev inequality equation}
    \mathbb{P}\Big([\eta_i^*;q^{(1)}_i;\cdots;q^{(z_i)}_i] \models SCP\Big) \geq 1 - \beta_{1,i}
    \end{equation}
    for some $\beta_{1,i} \geq \frac{\hat{M}^{(i)}}{\delta_{1,i}^2 \hat{N}_{1}}$.
\end{lemma}
\begin{proof}
Let us define the variance $\sigma^2(x,u):=Var(\frac{1}{\hat{N}_1}\sum_{k=1}^{\hat{N}_1}\hat{f}_i(x,u,w_k))$. Let $\hat{\sigma}^2:=\sup_{ x \in X, u \in U} \sigma^2(x,u)$. Under Assumption~\ref{variance assumption}, one can write the following set of inequalities $ \hat{\sigma}^2=\sup_{ x \in X, u \in U} \sigma^2(x,u) \leq \sup_{ x \in X, u \in U} \frac{\Tilde{M}(x,u)}{\hat{N}_1} \leq \frac{\hat{M}}{\hat{N}_1}$.

For a given $(x,u)\in (X \times U)$ pair, by employing Chebyshev's inequality and Markov's inequality, we have
    \begin{align*}
    &\mathbb{P}\biggl(\biggl||\mathbb{E}_w\hat{f}_i(x,u,w)\hspace{-0.2em}-\hspace{-0.2em}{{f}_i}(x,u)|\hspace{-0.2em}-\hspace{-0.2em}\bigl|\frac{1}{\hat{N}_1}\hspace{-0.4em}\sum_{k=1}^{\hat{N}_1}\hspace{-0.2em}\hat{f}_i(x,u,w_k)\hspace{-0.2em}\\&-\hspace{-0.2em}{{f}_i}(x,u)\bigr| \biggr| \hspace{-0.2em}\leq\hspace{-0.2em} \delta_{1,i}\biggr)= \mathbb{P}\biggl(\biggl||\mathbb{E}_w\hat{f}_i(x,u,w)\hspace{-0.2em}-\hspace{-0.2em}{{f}_i}(x,u)|\\&-\bigl|\frac{1}{\hat{N}_1}\sum_{k=1}^{\hat{N}_1}\hat{f}_i(x,u,w_k)-{{f}_i}(x,u)\bigr| \biggr|^2 \leq \delta_{1,i}^2\biggr)\geq 1-\frac{\sigma^2(x,u)}{\delta_{1,i}^2}.
    \end{align*}
    By taking the supremum from both sides, we have
     \begin{align*}
    &\sup_{ x \in X, u \in U}\mathbb{P}\biggl(\biggl||\mathbb{E}_w\hat{f}_i(x,u,w)\hspace{-0.2em}-\hspace{-0.2em}{{f}_i}(x,u)|\\&-\bigl|\frac{1}{\hat{N}_1}\sum_{k=1}^{\hat{N}_1}\hat{f}_i(x,u,w_k)-{{f}_i}(x_r,u)\bigr| \biggr|\hspace{-0.2em} \leq\hspace{-0.2em} \delta_{1,i}\biggr) \hspace{-0.25em} \geq \hspace{-0.5em}\sup_{ x \in X, u \in U}\hspace{-1em} 1-\frac{\sigma^2(x,u)}{\delta_{1,i}^2}\\&\geq 1 -\frac{\hat{M}}{\delta_{1,i}^2\hat{N}_1}.
\end{align*}
Since the above inequality holds for the supremum value,  \emph{i.e., for all $x \in X, u \in U$}, this will hold as follows:
\begin{align*}
    &\mathbb{P}(\biggl||\mathbb{E}_w\hat{f}_i(x,u,w)\hspace{-0.2em}-\hspace{-0.2em}{{f}_i}(x,u)|-\bigl|\frac{1}{\hat{N}_1}\sum_{k=1}^{\hat{N}_1}\hat{f}_i(x,u,w_k)\\&-{{f}_i}(x,u)\bigr| \biggr| \leq \delta_{1,i}) \geq 1 -\frac{\hat{M}}{\delta_{1,i}^2\hat{N}_1} \geq 1 - \beta_{1,i},
\end{align*}
where $ \beta_{1,i} \geq \frac{\hat{M}^{(i)}}{\delta_{1,i}^2 \hat{N}_{1}}$. 
Now, in \eqref{eq: SCP_for_one_state_n}, one can see that the expectation in the second constraint of \eqref{eq: SCP_for_one_state} is replaced by the empirical mean, along with $\delta_{1,i}$. Let $[\eta_i^*;q^{(1)}_i;\cdots;q^{(z_i)}_i]$ represent the solution of \eqref{eq: SCP_for_one_state_n}.
This will lead to 
\begin{align*}
     \mathbb{P}\Big([\eta_i^*;q^{(1)}_i;\cdots;q^{(z_i)}_i] \models SCP\Big) \geq 1 - \beta_{1,i}.
\end{align*}
\end{proof}

To transfer the results from solving SCP-$\hat{N}_1$ to ROP in~\eqref{eq: ROP_for_one_state}, the following lemma is required, which directly follows from Lemma~\ref{Chebyshev bound theorem for fsim-fmath}.
\begin{lemma}
\label{big Chebyshev bound theorem for fsim-fmath}
    Consider the chosen $\hat{N}_1$ for SCP-$\hat{N}_1$. Let Assumption~\ref{variance assumption} hold with given $\hat{M}^{(i)}$. For a given $\delta_{2,i}>0$, one gets
    \begin{align}
    \label{big Chebyshev inequality equation}&\mathbb{P}\biggl(\bigl|\mathbb{E}_w[\hat{f}_i(x,u,w)-\hat{f}_i(x_r,u,w)]-\nonumber \\&\frac{1}{{\hat{N}_1}}\sum_{k=1}^{\hat{N}_1}[\hat{f}_i(x,u,w_k)-\hat{f}_i(x_r,u,w_k)]\bigr|\leq \hspace{-0.15em}{\delta}_{2,i}\hspace{-0.15em} \biggr)\geq 1-\beta_{2,i}
\end{align}
with ${{\beta}}_{2,i} \geq \frac{2\hat{M}^{(i)}}{\delta_{2,i}^2 \hat{N}_1}$ for all $x \in X$, $x_r \in X_r$ such that, $\lVert x-x_r\rVert \leq \epsilon,$ for all $r \in \{1,2,\ldots,N\}$, for all $u \in U$.
\end{lemma}
\begin{proof}
    The above equation is true from the fact that $Var(X-Y) \leq 2Var(X)$ if $X$ and $Y$ are i.i.d. random variables of the same probability space. Using the inequality $Var\bigl(\hat{f}_i(x,u,w)-\hat{f}_i(x_r,u,w)\bigr) \leq 2\hat{M}^{(i)}$ and following a similar procedure of proof of Lemma~\ref{Chebyshev bound theorem for fsim-fmath}, one would be able to obtain the result in~\eqref{big Chebyshev inequality equation}.
\end{proof}

\section{Formal Quantification of Stochastic Simulation-Gap Function}
In this section, we formally quantify the stochastic simulation-gap function. To achieve this, we make the following assumption.

\begin{assumption}
\label{big Lipschitz continuity assumption}
    Let the empirical mean of the random variable $\hat{f}_i(x,u,w)$, the deterministic functions $f_i(x,u)$, and $q_i^\top p_i(x,u)$ be Lipschitz continuous in $x$ with Lipschitz constants $\hat{\mathcal{L}}_{\hat{f}}^{(i)}$, $\mathcal{L}_f^{(i)}$, and $\mathcal{L}_{2}^{(i)}$, respectively.
\end{assumption}   
\begin{remark}
    Assumption~\ref{big Lipschitz continuity assumption} is not restrictive, as the Lipschitz constant can be estimated from the data \cite{Wood1996EstimationOT,10066195,samari2024data}.
\end{remark}

Under Assumption~\ref{big Lipschitz continuity assumption}, the following theorem establishes a guarantee that the solution of SCP-$\hat{N}_1$ remains valid for the original ROP \eqref{eq: ROP_for_one_state} with a specified confidence level.
\begin{theorem}
\label{main_theorem}
    Consider Assumption \ref{big Lipschitz continuity assumption}, the $\hat{N}_1$ selected for SCP-$\hat{N}_1$, some arbitrary constant $\delta_{2,i}>0$, and the solution of SCP-$\hat{N}_1$, $q_i=[q_i^{(1)};\cdots;q_i^{(z_i)}]$ satisfying Lemma~\ref{Chebyshev bound theorem for fsim-fmath}. Then, for all $i \in I$, for the mathematical model in \eqref{disc_mathematical_model} and the simulator model in \eqref{simulator_model}, one has
\begin{equation}
\mathbb{P}\Bigl(\bigl|\mathbb{E}_{w}\hat{f}_i(x,u,w)-{f}_i(x,u)\bigr| \leq\gamma_i(x,u) \Bigr) \geq  1-{\beta}_{1,i}-\beta_{2,i},
\label{big actual_gamma_with probability result}
\end{equation}
for all $x \in X$, for all $u \in U$, for all $w \in V_w$,
where \begin{equation}
\label{actual_gamma}
    \gamma_i(x,u):=\mathcal{L}_{x}^{(i)} \epsilon +q_i^\top p_i(x,u)+{\delta}_{2,i}
\end{equation}
with $\mathcal{L}_{x}^{(i)}=\mathcal{L}_f^{(i)}+\hat{\mathcal{L}}_{\hat{f}}^{(i)}+\mathcal{L}_{2}^{(i)}$, $\epsilon>0$ as defined in \eqref{max_epsilon}, and $\beta_{1,i}$ and $\beta_{2,i}$ defined as in \eqref{Chebyshev inequality equation} and \eqref{big Chebyshev inequality equation}.
\end{theorem}
\begin{proof}
For a given $i\in I$, for all $x\in X$, for any $x_r\in X_r$ such that $\|x-x_r\|\leq\epsilon$, and for all $u\in U$, we have
     \begin{align*}
     &|\mathbb{E}_{w}\hat{f}_i(x,u,w)-{f}_i(x,u)|= |\mathbb{E}_w\hat{f}_i(x,u,w)-{f}_i(x,u)\\&   -\mathbb{E}_w\hat{f}_i({x}_r,{u},{w})+f_i({x}_r,{u})+ \mathbb{E}_w\hat{f}_i({x}_r,{u},{w})-{f}_i({x}_r,{u})|\\
     & \leq |{f}_i(x,u)-f_i({x}_r,{u})|+|\mathbb{E}_w\hat{f}_i(x,u,w)-\mathbb{E}_w\hat{f}_i({x}_r,{u},{w})|\\&\quad +|\mathbb{E}_w\hat{f}_i({x}_r,{u},{w})-{f}_i({x}_r,{u})|.
\end{align*}
Using the result of Lemma~\ref{Chebyshev bound theorem for fsim-fmath}, Lemma~\ref{big Chebyshev bound theorem for fsim-fmath}, and Assumption~\ref{big Lipschitz continuity assumption}, the above inequality is reformulated as follows, holding with a probability that is quantified later:
\begin{align*}
&|\mathbb{E}_{w}\hat{f}_i(x,u,w)-{f}_i(x,u)|
\leq |{f}_i(x,u)-f_i({x}_r,{u})|+\\&\biggl|\frac{1}{\hat{N}_1}\hspace{-0.4em}\sum_{k=1}^{\hat{N}_1}[\hat{f}_i(x,u,w_k)\hspace{-0.2em}-\hspace{-0.2em}\hat{f}_i(x_r,u,w_k)]\biggr|\hspace{-0.2em}+\hspace{-0.2em}\delta_{2,i}+q_i^\top p_i(x_r,u)\\
    & \leq \mathcal{L}_f^{(i)}\lVert x\hspace{-0.2em}-\hspace{-0.2em}x_r \rVert  +\hat{\mathcal{L}}_{\hat{f}}^{(i)}\lVert x\hspace{-0.2em}-\hspace{-0.2em}{x}_r \rVert\hspace{-0.2em}+\hspace{-0.2em}q_i^\top p_i({x}_r,{u}_r)\hspace{-0.2em}-\hspace{-0.2em}q_i^\top p_i(x,u)+\\& \quad \quad q_i^\top p_i(x,u)+{\delta}_{2,i}\\
     & \leq \mathcal{L}_f^{(i)}\lVert x\hspace{-0.1em}-\hspace{-0.1em}x_r \rVert  \hspace{-0.2em}+\hat{\mathcal{L}}_{\hat{f}}^{(i)}\lVert x\hspace{-0.2em}-\hspace{-0.2em}{x}_r \rVert\hspace{-0.2em}+ \hspace{-0.2em}\mathcal{L}_{2}^{(i)}\lVert x\hspace{-0.2em}-\hspace{-0.2em}{x}_r \rVert
\hspace{-0.2em}+\hspace{-0.2em}q_i^\top p_i(x,u)+{\delta}_{2,i}\\
      &\leq \mathcal{L}_{x}^{(i)} \epsilon +q_i^\top p_i(x,u)+{\delta}_{2,i}:=\gamma_i(x,u),
\end{align*}
where $\mathcal{L}_{x}^{(i)}= \mathcal{L}_f^{(i)}+\hat{\mathcal{L}}_{\hat{f}}^{(i)}+\mathcal{L}_{2}^{(i)}$. The above result holds with a certain probability. To quantify this, let us define the following events as $\mathcal{C}_1=\Big\{\bigl||\mathbb{E}_w\hat{f}_i(x,u,w)-\mathbb{E}_w\hat{f}_i(x_r,u,w)|-\frac{1}{\hat{N}_1}\sum_{k=1}^{\hat{N}_1}|\hat{f}_i(x,u,w_k)-\hat{f}_i(x_r,u,w_k)|\bigr| \leq {\delta}_{2,i}\Big\},\mathcal{C}_2=\Big\{ \bigl|\mathbb{E}_{{w}}\hat{f}_i(x_r,u,{w})-f_i(x_r,u)\bigr|\leq q_i^\top p(x_r,u)\Big\},\mathcal{C}_3=\Big\{\bigl|\mathbb{E}_{w}\hat{f}_i(x,u,w)-{f}_i(x,u)\bigr|  \leq \gamma_i(x,u)\Big\}.$
Since $\mathcal{C}_1 \cap \mathcal{C}_2 \subseteq \mathcal{C}_3$, it follows that $\mathbb{P}(\mathcal{C}_3) \geq \mathbb{P}(\mathcal{C}_1 \cap \mathcal{C}_2)$. By Lemma~\ref{Chebyshev bound theorem for fsim-fmath}, $\mathbb{P}(\mathcal{C}_2) \geq 1 - \beta_{1,i}$, and by \eqref{big Chebyshev inequality equation}, $\mathbb{P}(\mathcal{C}_1) \geq 1 - \beta_{2,i}$. Applying the union bound $\mathbb{P}(\mathcal{C}_1 \cup \mathcal{C}_2) \leq \mathbb{P}(\mathcal{C}_1)+\mathbb{P}(\mathcal{C}_2)$, we obtain $\mathbb{P}(\mathcal{C}_1 \cap \mathcal{C}_2) \geq 1-\mathbb{P}(\mathcal{C}_1^c) -\mathbb{P}(\mathcal{C}_2^c)\geq 1-{\beta}_{1,i}-\beta_{2,i}$. Hence, $\mathbb{P}(\mathcal{C}_3) \geq 1 - \beta_{1,i} - \beta_{2,i}$, and accordingly, we have
$\mathbb{P}\bigl(\bigl|\mathbb{E}_{w}\hat{f}_i(x,u,w)-{f}_i(x,u)\bigr| \leq\gamma_i(x,u) \bigr) \geq  1-{\beta}_{1,i}-\beta_{2,i}$.
\end{proof} 
The above result gives a probabilistic guarantee for the $i$-th state; the next theorem extends it to the full system dynamics.
\begin{theorem}
\label{total_probability_result_theorem}
Given the result of Theorem~\ref{main_theorem} for the mathematical model in \eqref{disc_mathematical_model} and the simulator model in \eqref{simulator_model}, the following result holds true:
\begin{align}
\label{controller_design_equation}
    \mathbb{P}\biggl(\Big\{\mathbb{E}_{w} \hat{f}(x,u,w)&\in f(x,u)+[-\gamma(x,u),\gamma(x,u)]\Big\}\biggr)\nonumber \\& \geq 1-\sum_{i=1}^n(\beta_{1,i}+\beta_{2,i}).
\end{align}
\end{theorem}
\begin{proof}
For a given $i \in I$, $x \in X$, $u \in U$, and $w \in V_w$, let us define the following events:
$\mathcal{E}_i=\Big\{\mathbb{E}_{w}\hat{f}_i(\hspace{-0.1em}x\hspace{-0.1em},\hspace{-0.1em}u\hspace{-0.1em},\hspace{-0.1em}w\hspace{-0.1em}) \hspace{-0.2em}\in\hspace{-0.2em} {f}_i(\hspace{-0.1em}x\hspace{-0.1em},\hspace{-0.1em}u\hspace{-0.1em})+[-\gamma_i(\hspace{-0.1em}x\hspace{-0.1em},\hspace{-0.1em}u\hspace{-0.1em}),\gamma_i(\hspace{-0.1em}x\hspace{-0.1em},\hspace{-0.1em}u\hspace{-0.1em})]\Big\}$, 
$\mathcal{E}=\bigcap_{i=1}^n\Big\{\mathbb{E}_{w}\hat{f}_i(\hspace{-0.1em}x\hspace{-0.1em},\hspace{-0.1em}u\hspace{-0.1em},\hspace{-0.1em}w\hspace{-0.1em}) \hspace{-0.2em}\in\hspace{-0.2em} {f}_i(\hspace{-0.1em}x\hspace{-0.1em},\hspace{-0.1em}u\hspace{-0.1em})+[-\gamma_i(\hspace{-0.1em}x\hspace{-0.1em},\hspace{-0.1em}u\hspace{-0.1em}),\gamma_i(\hspace{-0.1em}x\hspace{-0.1em},\hspace{-0.1em}u\hspace{-0.1em})]\Big\}$.
It is clear that $\mathcal{E}=\Big\{\mathbb{E}_{w}\hat{f}(\hspace{-0.1em}x\hspace{-0.1em},\hspace{-0.1em}u\hspace{-0.1em},\hspace{-0.1em}w\hspace{-0.1em}) \hspace{-0.2em}\in\hspace{-0.2em} {f}(\hspace{-0.1em}x\hspace{-0.1em},\hspace{-0.1em}u\hspace{-0.1em})+[-\gamma(\hspace{-0.1em}x\hspace{-0.1em},\hspace{-0.1em}u\hspace{-0.1em}),\gamma(\hspace{-0.1em}x\hspace{-0.1em},\hspace{-0.1em}u\hspace{-0.1em})]\Big\}$.
From Theorem~\ref{main_theorem}, for all $i \in I$, we have 
$\mathbb{P}(\mathcal{E}_i)\geq 1-\beta_{1,i}-\beta_{2,i}$. We now compute the concurrent occurrences of events $\mathcal{E}_i$, namely $\mathbb{P}(\cap_{i=1}^n \mathcal{E}_i)$, as $
    \mathbb{P}(\mathcal{E}) =\mathbb{P}(\cap_{i=1}^n \mathcal{E}_i)=1-\mathbb{P}(\cup_{i=1}^n \mathcal{E}_i^c) \geq 1-\sum_{i=1}^n\mathbb{P}(\mathcal{E}_i^c) = 1-\sum_{i=1}^n(\beta_{1,i}+\beta_{2,i})$,
which completes the proof.
\end{proof}
Following Theorem \ref{total_probability_result_theorem}, the computed $\gamma(x,u)$ is statistically guaranteed via Chebyshev’s inequality to be valid with some confidence without requiring repeated trials, aligning with standard practice in learning-based control frameworks. Moreover, in our approach, one can notice that the soundness of the result is ensured through the probabilistic guarantees provided in Theorem \ref{total_probability_result_theorem}, and the completeness is inherently ensured by the construction of the framework, as seen in Remark \ref{gamma_triviality_remark}.
\begin{remark}
    The simulation gap $\gamma_i(x,u)$, which provides formal guarantees over the entire state space, depends on parameters such as $\delta_{2,i}$, $\epsilon$, and $\mathcal{L}_x^{(i)}$, as in \eqref{actual_gamma} and holds with probability at least $1 - \beta_{1,i} - \beta_{2,i}$ as in \eqref{big actual_gamma_with probability result}. Its conservativeness can be reduced by choosing a smaller $\epsilon$ (denser state-space sampling, resulting in a higher number of constraints in the SCP) and a smaller $\delta_{2,i}$, which can be achieved by increasing $\hat{N}_1$ for a fixed variance $\hat{M}^{(i)}$ and confidence level $\beta_{2,i}$. Similarly, for a desired confidence $\beta_{1,i}$, the bound associated with $\delta_{1,i}$ can be tightened by increasing $\hat{N}_1$. Although a larger $\hat{N}_1$ implies more data collection, simulators allow parallel sampling, making the process efficient given adequate computational resources.
\end{remark}

\subsection{Controller Synthesis}
\label{controller synthesis section}
This subsection presents controller design using the mathematical model and the stochastic simulation-gap function in \eqref{actual_gamma}, ensuring specification satisfaction in the simulator with probabilistic guarantees.
Using Theorem~\ref{total_probability_result_theorem}, the simulator dynamics (in its first moment) is represented as
\begin{equation}
    \label{basic_bound_equation}
   { \mathbb{E}_w\hat{f}(x,u,w) \in {f}(x,u)+[-\gamma(x,u),\gamma(x,u)]},
\end{equation}
which holds probabilistically, as given in \eqref{controller_design_equation}. The right-hand-side expression can be viewed as an uncertain system, a well-studied topic in control design. Examples include adaptive control \cite{singh2022adaptive}, sliding mode controller \cite{kim2000designing,utkin2013sliding}, and symbolic control \cite{tabuada2009verification}. In our work, we have used symbolic controllers \cite{Reissig_2017} for our case studies, utilizing the SCOTS toolbox \cite{rungger2016scots} {to design a controller for the system. Using Theorem~\ref{total_probability_result_theorem}, we can guarantee that the expected trajectory will satisfy the specification}. Symbolic abstraction techniques work by approximating concrete systems with symbolic models, employing automata-theoretic methods for satisfying complex desired specifications. Implementation details are omitted for brevity; see \cite{Reissig_2017,rungger2016scots} for more details.

\section{Case Studies}
\label{case studies subsection}
We validate the proposed data-driven approach on a TurtleBot and a nonlinear pendulum using MATLAB for the mathematical model and Gazebo/PyBullet for high-fidelity stochastic simulations. The experiments were performed on an AMD Ryzen 9 5950X with 128 GB RAM.

\textbf{A. TurtleBot.}
The first case study considers a TurtleBot, mathematically modelled by a deterministic system as: $x_1^+=x_1+\tau u_1\cos{x_3}$, $x_2^+=x_2+\tau u_1\sin{x_3}$, $x_3^+=x_3+\tau u_2$,
where $x_1$, $x_2$, and $x_3$ denote the position in the x and y axes and the orientation of the bot, respectively. The parameters $u_1$ and $u_2$ denote the linear and angular input velocities of the bot, respectively. The sampling time is chosen as $\tau=0.01s$. The TurtleBot model in a Gazebo simulator (stochastic simulator model) is shown in Fig. \ref{fig:  Turtlebot gazebo}. The state and finite input sets are, respectively, considered as $X=[0,4]\times[0,4]\times[-\pi,\pi]$ and $U=\{-1,-0.9,\ldots,0.9,1\} \times \{-1,-0.9,\ldots,0.9,1\}$. 
\begin{figure}
    \centering
    \hspace{-.2em}\includegraphics[scale=0.1]{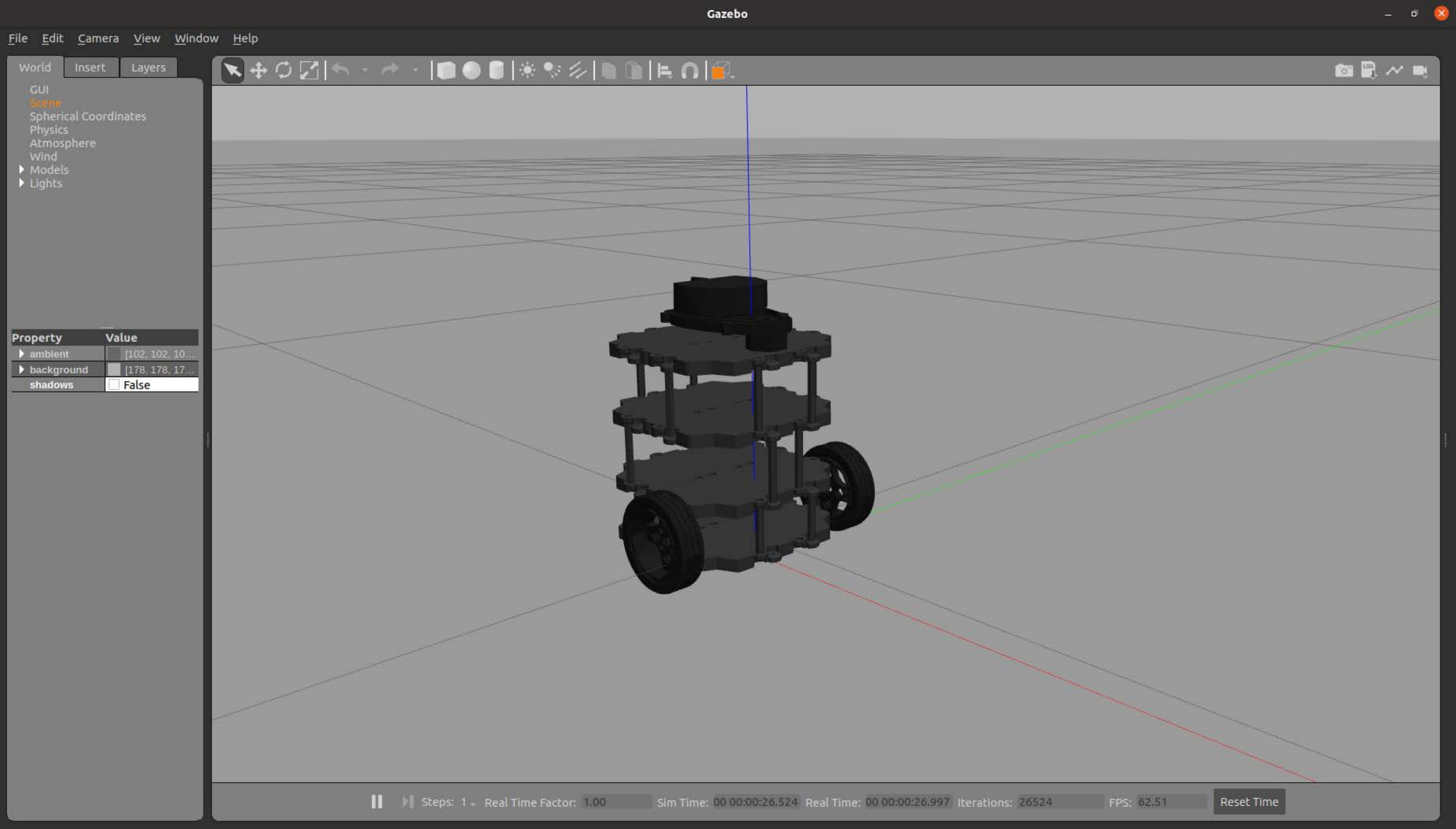}
    \caption{TurtleBot model in the Gazebo simulator.}
    \label{fig: Turtlebot gazebo}
\end{figure}
\begin{figure}
    \centering
    \includegraphics[width=1\linewidth]{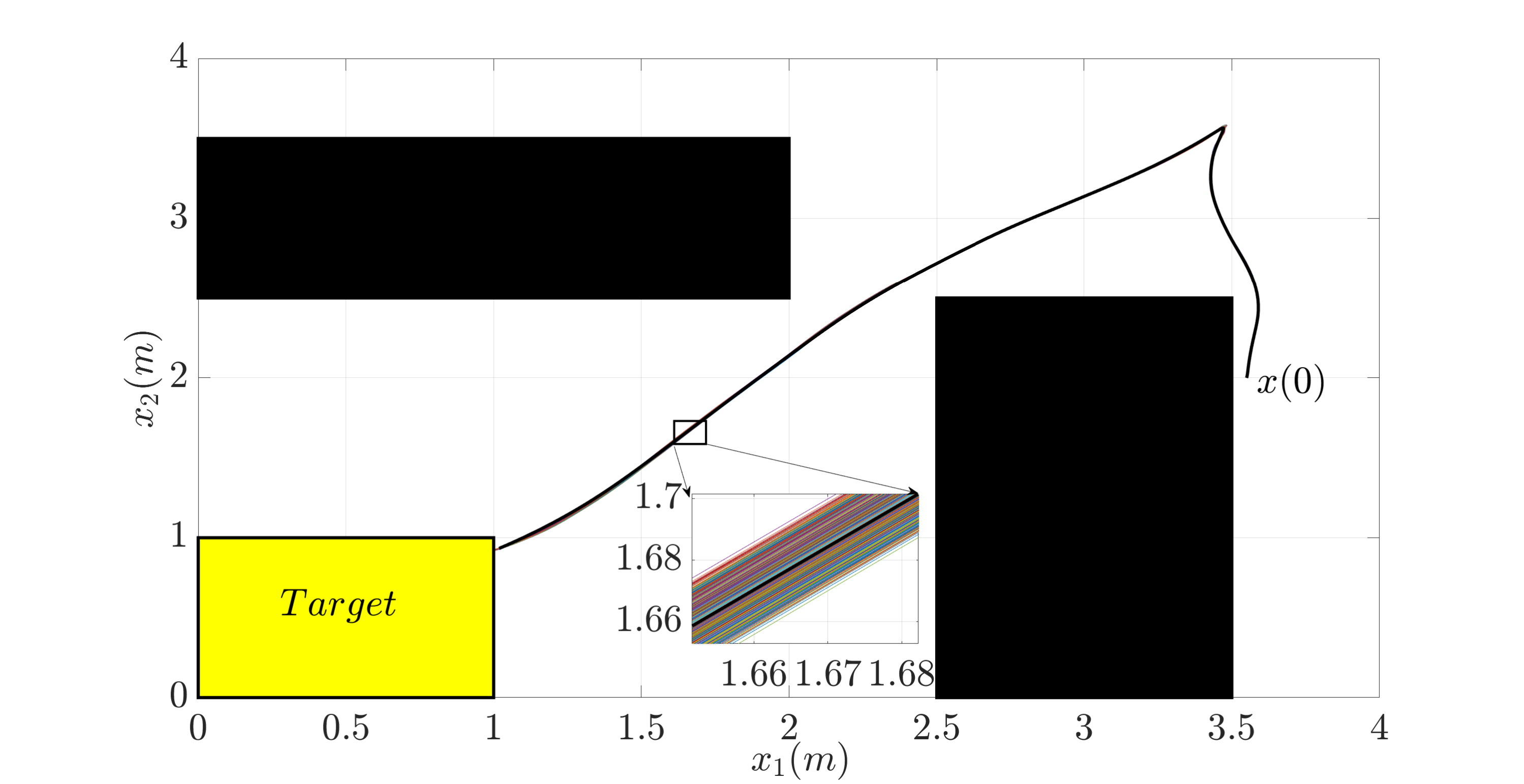}
    \caption{The realization of 1000 state trajectories from the initial state $(x,y,\theta) = (3.5,2,1.5)$. Black regions indicate obstacles. The controller designed using the mathematical model, along with the stochastic simulation gap, ensures the satisfaction of the specification in the stochastic simulator for all 1000 trajectories, even though the theorem guaranteed satisfaction in the first-moment sense. The dark black line shows the average trajectory.}
    \label{fig: turtlebot trajectories}
\end{figure}
For each sampled pair $(x_r,u)$ using $\epsilon=0.017$ as in \eqref{max_epsilon}, we run the mathematical model and the simulator from $x(0):=x_r$ with constant input $u$ for a duration $\tau$,  recording the state evolution at $\tau$, for all $r \in \{1,2,\ldots,N\}$. This process is repeated for $\hat{N}_1=1000$ in the simulator to capture stochastic variations. The value of $\hat{M}^{(1)}=\hat{M}^{(2)}=0.007,\hat{M}^{(3)}=0.0101$ which is calculated using empirical variance estimation method \cite{maurer2009empirical} and a conservative global bound is taken (\emph{e.g.} 10 times). We fix $\delta_{1,1}=0.01$, $\delta_{2,1}=0.01$, $\delta_{1,2}=0.01$, $\delta_{2,2}=0.01$, $\delta_{1,3}=0.05$, $\delta_{2,3}=0.05$ and obtain $\beta_{1,1}=0.0028$, $\beta_{2,1}=0.0056$, $\beta_{1,2}=0.0028$, $\beta_{2,2}=0.0056$, $\beta_{1,3}=0.0079$, $\beta_{2,3}=0.0158$ (according to Lemma~\ref{Chebyshev bound theorem for fsim-fmath} and Lemma~\ref{big Chebyshev bound theorem for fsim-fmath}). The structure of $\gamma_i(x,u)$ was fixed as $q_i^{(1)}x_1+q_i^{(2)}x_2+q_i^{(3)}x_3+q_i^{(4)}u_1+q_i^{(5)}u_2+q_i^{(6)}$, $i\in \{1,2,3\}$. We solve SCP-$\hat{N}_1$ in \eqref{eq: SCP_for_one_state_n} and obtain $q_1^{(4)}=0.00081,q_1^{(6)}=0.01,q_2^{(4)}=0.0001747,q_2^{(6)}=0.01,q_3^{(3)}=0.00011,q_3^{(5)}=-0.000145,q_3^{(6)}=0.0301$, and $q_1^{(1)}=q_2^{(1)}=q_3^{(1)}=q_1^{(2)}=q_2^{(2)}=q_3^{(2)}=q_1^{(3)}=q_2^{(3)}=q_1^{(4)}=q_2^{(4)}=q_1^{(5)}=q_2^{(5)} \approx 0$. We estimate the Lipschitz constants $\mathcal{L}_f^{(1)}=0.0151,\hat{\mathcal{L}}_{\hat{f}}^{(1)}=0.0856,\mathcal{L}_2^{(1)}=0.0009,\mathcal{L}_f^{(2)}=0.01,\hat{\mathcal{L}}_{\hat{f}}^{(2)}=0.1258,\mathcal{L}_2^{(2)}=0.0002,\mathcal{L}_f^{(3)}=1.003,\hat{\mathcal{L}}_{\hat{f}}^{(3)}=1.1675,\mathcal{L}_2^{(3)}=0.0001453$ according to \cite{10066195}. We then design simulation-gap functions $\gamma_1(x,u)=0.00081u_1+0.112$, $\gamma_2(x,u)=0.0001747u_1+0.112$, and $\gamma_3(x,u)=-0.0001x_3-0.0001u_2+0.168$, which individually hold with a probability of $0.9916$, $0.9916$, $0.9763$ (in the sense of Theorem~\ref{main_theorem}). Incorporating $\gamma(x,u)$ into the controller synthesis ensures the expected trajectory in the stochastic simulator reaches the target while avoiding obstacles, satisfying the desired specification with a $95.95\%$ guarantee (as per Theorem~\ref{total_probability_result_theorem}). Unlike \cite{sangeerth2024towards}, which led to 44 collisions in 1000 runs (see Fig.~\ref{fig: with det gamma failure}), our method accounts for simulator stochasticity, resulting in zero collisions across all runs, which can be seen in Fig.~\ref{fig: turtlebot trajectories}. Data collection took around 8.7 hours, and the optimization completed in 5 seconds. The data collection step could be parallelized based on available resources.

\textbf{B. Pendulum System.} The second case study considers a pendulum system whose deterministic mathematical model is given as follows: $x_1^+=x_1+\tau x_2$, $x_2^+= -\frac{3g\tau}{2l}\sin{x_1}+x_2+\frac{3\tau u}{ml^2}$,
\begin{figure}
\centering
  \includegraphics[scale=0.2]{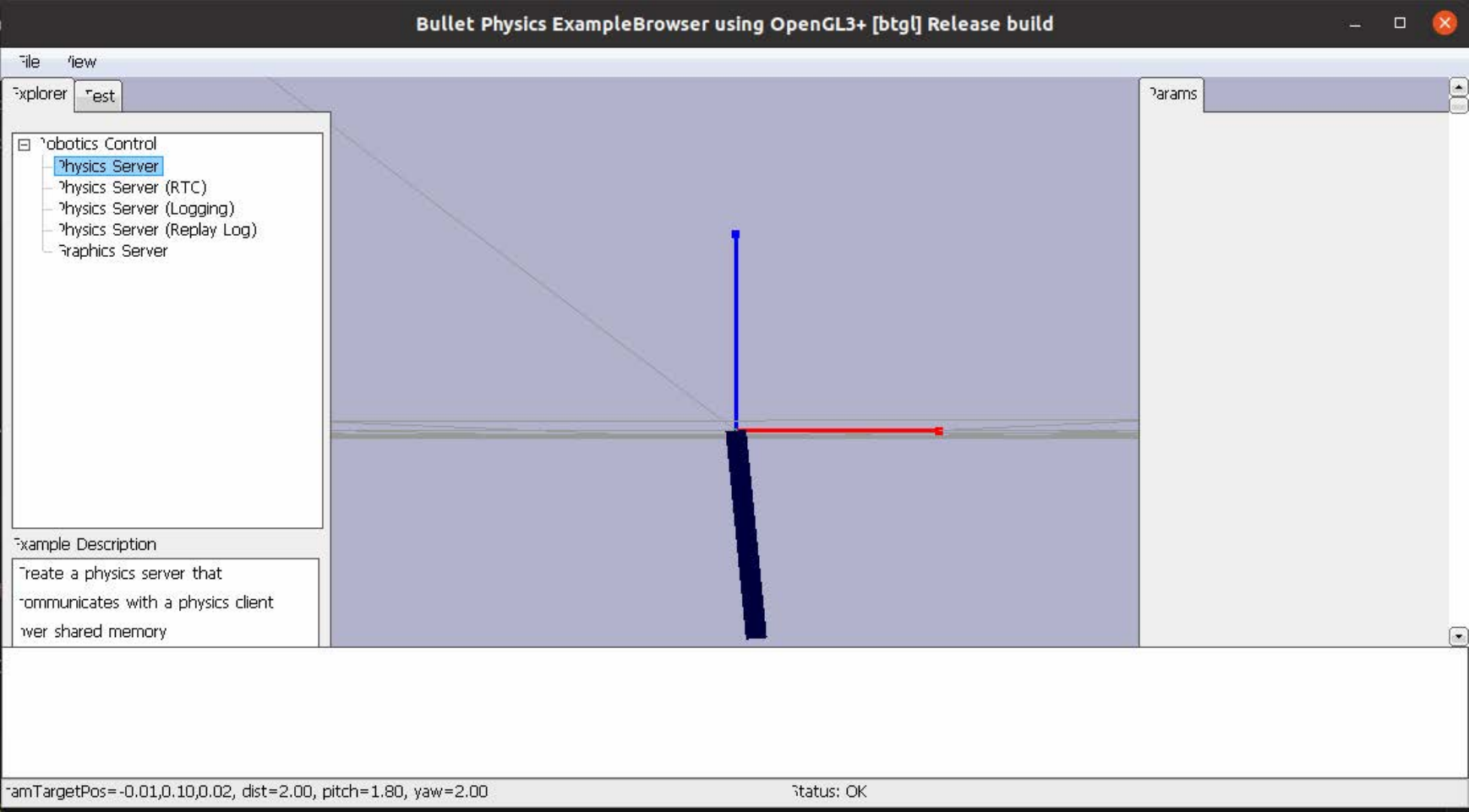}
  \caption{Pendulum model in the PyBullet simulator.}
  \label{fig: Pendulum model in Py-Bullet simulator}
\end{figure}
where $x_1$, $x_2$, and $u$ are the angular position, angular velocity, and torque input, respectively. The parameters $m=1kg$, $g=9.81m/s^2$, and $l=1m$ are the pendulum's mass, acceleration due to gravity, and rod length. The pendulum model simulated in the PyBullet stochastic high-fidelity simulator PyBullet is depicted in Fig. \ref{fig: Pendulum model in Py-Bullet simulator}. The parameter $\tau=0.005$ is the sampling time. The state-space is $X=[-0.2,0.2] \times[-0.5,0.5]$, while the input set is considered as $U=\{-1,-0.9,\ldots,0.9,1\}$. 
Data are collected using the procedure discussed in Section \ref{Sec:Data-driver Framework} with state-space discretization $\epsilon=0.004$. The stochastic simulation-gap functions are quantified as 
\begin{figure}
    \centering
    \includegraphics[width=0.9\linewidth]{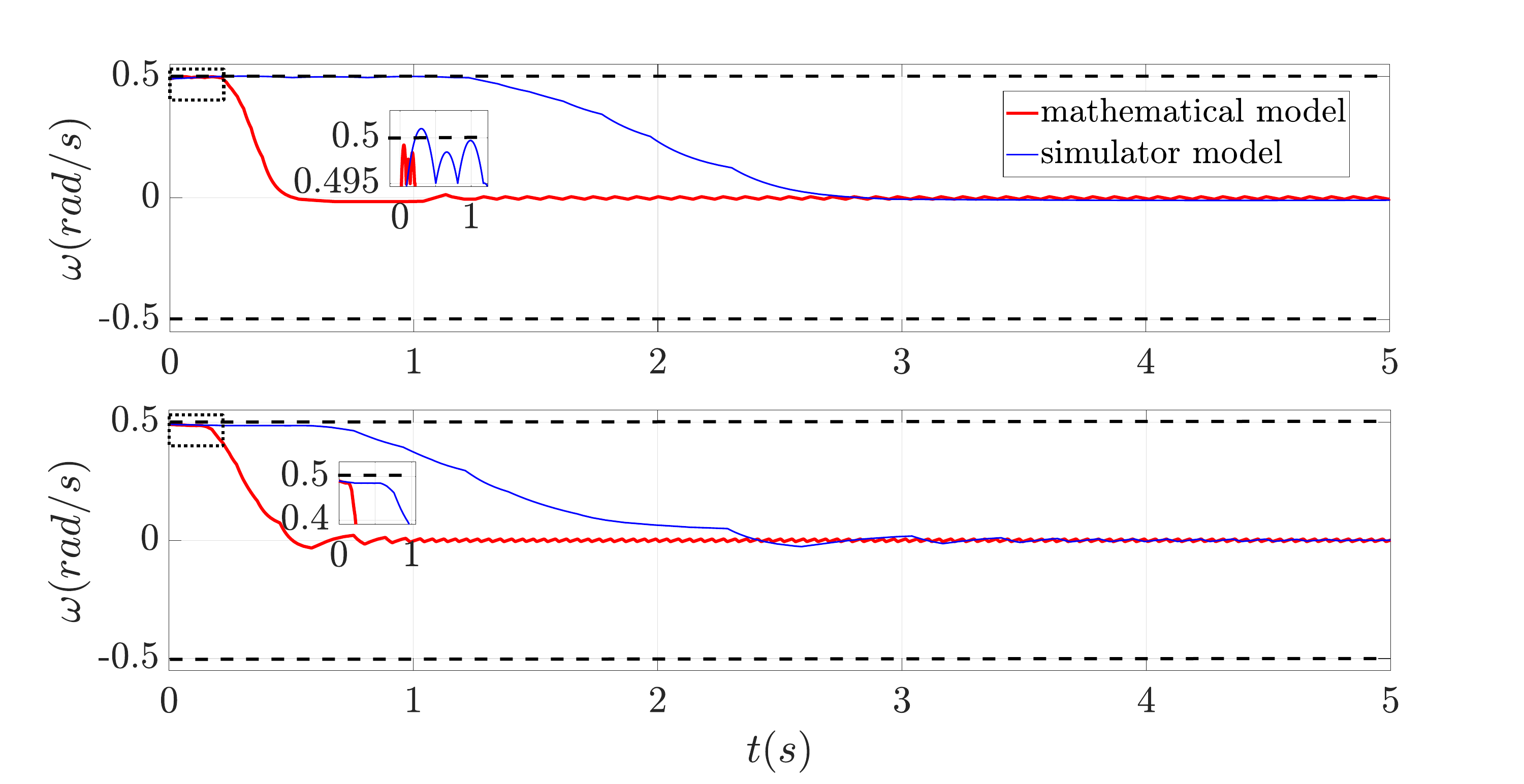}
    \caption{The mathematical system trajectory is shown in red, and the stochastic simulator trajectory in blue. In PyBullet, the controller synthesized using $\gamma(x,u)$ from \cite{sangeerth2024towards} violates the invariance specification (top). In contrast, the controller synthesized using the proposed approach ensures invariance in PyBullet in the first-moment sense (bottom), based on an average of 1000 trajectories. The $x_1$ trajectory, which remained within the safe set in both cases, is omitted for brevity. Multiple runs showing stochasticity are also omitted for brevity.}
    \label{fig:pendulum plots}
\end{figure}
$\gamma_1(x,u)=0.001x_2^2+0.00011x_2u+0.01189$ and 
$\gamma_2(x,u)=0.062x_1^2-0.0032x_2^2-0.0043u^2-0.0086x_1x_2-0.041x_1u-0.0002x_2u-0.0045x_1+0.00025x_2+0.1187$,
which holds with a guarantee of $0.973$ and $0.9649$, respectively (in the sense of Theorem~\ref{main_theorem}).
We use the SCOTS toolbox \cite{rungger2016scots} to design a symbolic controller for the invariance specification with safe set $[0,0.2] \times [-0.5,0.5]$. In Fig.~\ref{fig:pendulum plots} (top), the controller designed using the deterministic $\gamma(x,u)$ from \cite{sangeerth2024towards} meets the specification in the mathematical model but fails in the simulator due to unmodeled stochasticity. We address this by re-synthesizing the controller using our proposed stochastic simulation-gap function, which incorporates simulator stochasticity. As shown in Fig.~\ref{fig:pendulum plots} (bottom), the new controller satisfies the specification in the simulator (in the first-moment sense) with a $93.8\%$ probabilistic guarantee, as per Theorem~\ref{total_probability_result_theorem}. Data collection took around 0.4 hours, and the optimization completed in 2 seconds.

\section{Conclusion and Future Work}
This letter introduces a formal approach to quantify the gap between a nominal mathematical model and a high-fidelity \emph{stochastic} simulator, offering guarantees that account for simulator stochasticity, unlike existing approaches. When the high-fidelity simulator closely approximates real-world dynamics, our framework enables reliable controller synthesis, facilitating seamless deployment in real systems.  While we use the expected trajectory as a practical approximation under uncertainty, we aim to enhance robustness with respect to the true trajectory in future work. Currently, for our case studies, we use empirical variance estimation to estimate $\hat M^{(i)}$ (of Assumption \ref{variance assumption}), and we plan to explore other methods like sub-Gaussian analysis \cite{wainwright2019high} and Popoviciu’s inequality \cite{popoviciu1935equations} in future work. Future work will also explore measuring this gap directly using real-world data.

\bibliographystyle{ieeetr}
\bibliography{sources}

\end{document}